\documentclass[11pt,a4paper]{article}
\usepackage{amsfonts,amssymb,amsthm}
\usepackage{cite} 
\usepackage[T2A]{fontenc}
\usepackage[cp1251]{inputenc}
\usepackage[english]{babel}
\usepackage{amsmath}
\usepackage{amsfonts, dsfont, graphicx, srcltx, color}

\usepackage[unicode]{hyperref}
\usepackage[left=1.5cm,right=1.5cm,    top=2cm,bottom=3cm,bindingoffset=0cm]{geometry}

\newtheorem{theorem}{Theorem}[section]

\newtheorem{hypothesis}{Hypothesis}[section]
\newtheorem{corollary}{Corollary}[section]

\def\const {\hbox{const}}

\allowdisplaybreaks
\numberwithin{equation}{section}

\title{\bf On the construction of solutions of the Davey--Stewartson I equation using an open Toda chain}
\author{\bf I.T.~Habibullin, A.R.~Khakimova }

\date{} 

\begin{document}
\maketitle

%Institute of Mathematics, Ufa Scientific Center, Russian Academy of Sciences \\

%Bashkir State University, Ufa, Russia

\abstract {An effective method for constructing explicit solutions to the Davey--Stewartson type integrable equations is discussed based on the use of a dressing chain. The application of the method is exemplified by the equation DS I, for which a new class of explicit solutions is constructed, containing freedom in two arbitrary functions. In this case the generalized Toda lattice corresponding to the simple Lie algebra $A_2$ is  used as a dressing chain.} 

\vspace{0.5cm}


\vspace{0.5cm}

\large

\section{Introduction}

The problem of constructing explicit particular solutions of the Davey--Stewartson I equation 
\begin{equation}\label{eq-DSI}
\begin{aligned}
&iu_t+\Delta u+uv=0,\\
&v_{\xi\eta}+\frac{1}{2}\Delta |u|^2=0
\end{aligned}
\end{equation}
is considered, where $\Delta$ stands for the Laplace operator $\frac{\partial^2}{\partial\xi^2}+\frac{\partial^2}{\partial\eta^2}$. For this purpose we use the concept of  dressing chain introduced years ago by A.B.~Shabat and R.I.~Yamilov in \cite{ShYAA91}, \cite{ShY97}.

The Davey--Stewartson equation has important applications in hydrodynamics and other areas of physics \cite{BenneyRoskes}, \cite{DaveyStewartson}, \cite{GrinevichSantini}. In fact, it is a spatially two-dimensional generalization of the well-known integrable nonlinear Schr\"odinger equation.  
Currently, interest in the class of (1+2)-dimensional Schr\"odinger type equations has increased significantly due to their important application in nonlinear optics (see \cite{MuradOmar}, \cite{SeadawyCheemaaBiswas}, \cite{Hosseini}, \cite{Bilal}). At the end of the eighties of the last century, soliton solutions of this kind equations were discovered (see \cite{BoitiLeonMartinaPempinelli}, \cite{BoitiLeonPempinelli}, \cite{FokasSantini89}, \cite{FokasSantini90}). Since then, the DS equation has been very actively studied by many authors (see, for instance, \cite{KonopelchenkoMatkarimov90}, \cite{Pogrebkov94}). In the last two decades, new effective approaches to the study of multidimensional integrable models of the DS type have been developed (see \cite{Pogrebkov21}, \cite{BogdanovKonopelchenko06}, \cite{FerapontovKhusnutdinovaPavlov05}, \cite{Taimanov21}, \cite{Kudryashov20} and references therein).

It was discovered in \cite{ShY97}, \cite{LShY93} that there is a deep relation between \eqref{eq-DSI} and the Toda lattice. More precisely, the two-dimensional Toda lattice  can be interpreted as an iteration of the B\"acklund transformation for a coupled system associated with the Davey--Stewartson equation (see {\eqref{sym-Toda}, \eqref{nonloc-Toda} below). In other words, the Toda lattice provides a dressing chain for the coupled system. 

Dressing chains of nonlinear integrable equations in $1+1$ dimensions are known as an effective tool for constructing particular solutions (see, e.g., \cite{VeselovShabat93}, \cite{DegasShabat94}). Dressing chains in $1+2$ dimensions, although discovered more than thirty years ago, have not found effective application for a long time due to problems with nonlocalities (see \cite{LShY93}). Some progress in this direction was made in our paper \cite{HabKhUMJ24}, where an effective way to overcome the problem of nonlocalities was proposed. The essence of our approach is that to construct a solution to an integrable partial differential equation we use not the entire infinite dressing chain, but only its part of finite length. In other words, we use a finite-field reduction of the chain obtained by imposing boundary conditions that preserve the integrability property. Later, in \cite{GarHabPhysD} this approach was successfully used to construct classes of explicit solutions to the well-known Ishimori equation.

In this paper, we adapted the reduced dressing chain method for a class of nonlinear equations of 1+2 dimensional NLS type. In the frame of this approach we   constructed new analytical solutions of the Davey--Stewartson equation I.

Let us briefly discuss how the paper is organized. In \S2 we recall the known fact  that Toda lattice, in suitable variables, is a dressing chain for a coupled system generalizing \eqref{eq-DSI} (see \cite{ShYAA91}, \cite{ShY97}). To find a solution to the coupled system, one can use the notion of Lax pair. However the standard Lax pairs of the coupled system and Toda lattice essentially differ, they are not consistent. Therefore, we are forced to first modify the Lax pair of the Toda lattice. In the third section we present a combined Lax pair for the dressing chain, which is compatible with the Lax pair for the coupled system. In \S4 some well-known results on open Toda lattices are briefly recalled. In \S5, explicit solutions of the generalized Toda lattice related to $A_2$ are presented and the eigenfunctions of one of the Lax operators are found. Here the problem of determining the dynamics on $t$ of explicit solutions of the chain $A_2$ is also solved. As a result, an explicit solution of the coupled system of nonlinear partial differential equations is found. In Section 6 we presented a new class of solutions to the DS I equation \eqref{eq-DSI} containing two arbitrary functions. In \S 7 we identified a specific representative of this class and plotted its graph.

\section{Problem statement} %and solution method}

It is well known that the two-dimensional Toda lattice 
\begin{equation}\label{Toda}
q_{n,xy}=e^{q_{n+1}-q_n}-e^{q_n-q_{n-1}}
\end{equation}
admits a generalized symmetry  of the following form (see, for instance \cite{ShY97}, \cite{UTakasaki}):
\begin{equation}\label{sym-Toda}
\begin{aligned}
&i\hat{q}_{1,t}+\hat{q}_{1,\xi\xi}+\hat{q}_{1,\eta\eta}+\hat{q}_{1}(g_1+g_2)=0, \\
-&i\hat{q}_{2,t}+\hat{q}_{2,\xi\xi}+\hat{q}_{2,\eta\eta}+\hat{q}_{2}(g_1+g_2)=0,
\end{aligned}
\end{equation}
where functions $g_1$ and $g_2$, called non-local variables, are defined by the  equations
\begin{equation}\label{nonloc-Toda}
\frac{\partial}{\partial \xi}g_1=-\frac{1}{2}\frac{\partial}{\partial \eta}(\hat{q}_{1}\hat{q}_{2}), \qquad
\frac{\partial}{\partial \eta}g_2=-\frac{1}{2}\frac{\partial}{\partial \xi}(\hat{q}_{1}\hat{q}_{2}).
\end{equation}

In order to explain in more detail the relation between models \eqref{Toda} and \eqref{sym-Toda}, \eqref{nonloc-Toda}, we first clarify the connection between the independent variables: $\xi=-2x$, $\eta=-2y$ included in these models and the unknown functions
\begin{equation}\label{connection}
\hat{q}_{1}(n)=-e^{-q_n}, \qquad \hat{q}_{2}(n)=e^{q_{n+1}}.
\end{equation}
In new variables, the Toda lattice takes the form of system \eqref{sys-q1q2-xi-eta} (see below). Then, as shown in the articles \cite{ShY97}, \cite{UTakasaki}, the phase flows defined by two systems \eqref{sys-q1q2-xi-eta} and \eqref{sym-Toda}, \eqref{nonloc-Toda} commute.

Following the results of \cite{ShYAA91} and \cite{ShY97}, we interpret the Toda lattice as a symmetry of system \eqref{sym-Toda}, \eqref{nonloc-Toda} with discrete time $n$ (or, synonymously, as a dressing chain). This means that shift in $n$ and evolution in $t$ are permutable on the common solutions of both systems: \eqref{sys-q1q2-xi-eta} and \eqref{sym-Toda}, \eqref{nonloc-Toda}. Actually, the shift in $n$ performed according to the rule
$$(q_n,q_{n+1})\longmapsto (q_{n+1},q_{n+2})$$
defines a B\"acklund transformation that maps the solution $(\hat{q}_1(n),\hat{q}_2(n))$ of system \eqref{sym-Toda}, \eqref{nonloc-Toda} into its new solution $(\hat{q}_1(n+1),\hat{q}_2(n+1))$. Since a shift of the variable $n$ is naturally related to the Toda lattice, it is easy to obtain an explicit expression for this B\"acklund transformation:
\begin{equation}\label{backlund}
\begin{aligned}
&\hat{q}_{1}(n+1)=-\frac{1}{\hat{q}_{2}(n)}, \\
&\hat{q}_{2}(n+1)=4\hat{q}_{2,\xi\eta}(n)-4\frac{\hat{q}_{2,\xi}(n)\hat{q}_{2,\eta}(n)}{\hat{q}_{2}(n)}\hat{q}^2_{2}(n)\hat{q}_{1}(n)
\end{aligned}
\end{equation}
due to \eqref{Toda} and \eqref{connection}.
Nonlocal variables $g_1(n+1)$ and $g_2(n+1)$ corresponding to $(\hat{q}_1(n+1),\hat{q}_2(n+1))$ are calculated according to the rule\begin{equation}\label{nonloc-new}
\begin{aligned}
&g_1(n+1)=g_1(n)+2\frac{\partial^2}{\partial \eta^2}\ln \hat{q}_2(n), \\
&g_2(n+1)=g_2(n)+2\frac{\partial^2}{\partial \xi^2}\ln \hat{q}_2(n).
\end{aligned}
\end{equation}
Since transformation \eqref{backlund}, \eqref{nonloc-new} is given by a differential substitution, it is not invertible by itself. However, the transformation will be invertible if we keep in mind that function $q_n$ in \eqref{connection} satisfies the Toda lattice equation (see \cite{LShY93}). The inverse transformation is expressed as
\begin{align}
&\begin{aligned}\label{inverse-trans-q}
&\hat{q}_1(n-1)=-4\hat{q}_{1,\xi\eta}+4\frac{\hat{q}_{1,\xi}(n)\hat{q}_{1,\eta}(n)}{\hat{q}_{1}(n)}-\hat{q}^2_{1}(n)\hat{q}_{2}(n), \\
&\hat{q}_2(n-1)=-\frac{1}{\hat{q}_{1}(n)},
\end{aligned}\\
&\begin{aligned}\label{inverse-trans-g}
&g_1(n-1)=g_1(n)+2\frac{\partial^2}{\partial \eta^2}\ln \hat{q}_1(n), \\
&g_2(n-1)=g_2(n)+2\frac{\partial^2}{\partial \xi^2}\ln \hat{q}_1(n).
\end{aligned}
\end{align}
Thus the problem is to find a common solution of coupled system  \eqref{sym-Toda}, \eqref{nonloc-Toda} and its symmetry  with the discrete time (see \eqref{sys-q1q2-xi-eta} below). Recall that the standard symmetry method for finding a particular solution of an equation is based on use of the stationary part of its symmetry. However, this approach is inefficient here due to problems with non-local variables. Therefore, we use here another idea -- we take advantage of a finite-field reduction of the lattice compatible with the integrability property.

\section{Compatible Lax pairs for the coupled system and Toda lattice}

It is well known that  system \eqref{sym-Toda}, \eqref{nonloc-Toda} can be represented as a compatibility condition for the following system of linear equations (see, for instance, \cite{Flora})
\begin{align}
&\begin{aligned}\label{Lax-psi-varphi-t}
&i\psi_{1,t}=\psi_{1,\xi\xi}-\psi_{1,\eta\eta}-g_1\psi_1+\hat{q}_{1,\xi}\varphi_2, \\
&i\varphi_{2,t}=\varphi_{2,\xi\xi}-\varphi_{2,\eta\eta}+g_2\varphi_2-\hat{q}_{2,\eta}\psi_1,
\end{aligned}\\
&\begin{aligned}\label{Lax-psi-varphi-xi-eta}
\psi_{1,\xi}=-\frac{1}{2}\hat{q}_{1}\varphi_2, \qquad
\varphi_{2,\eta}=-\frac{1}{2}\hat{q}_{2}\psi_1.
\end{aligned}
\end{align}
Recall that the standard Lax pair of the Toda lattice is of the form:
\begin{align} \label{Lax-Toda-1}
\psi_{n,x}=-e^{q_{n+1}-q_n}\psi_{n+1}, \qquad \psi_{n,y}=-q_{n,y}\psi_n+\psi_{n-1}.
\end{align}
By applying the gauge transformation
\begin{align*}
\psi_n=e^{-q_n}\varphi_n
\end{align*}
one gets from \eqref{Lax-Toda-1} another Lax pair  for the Toda lattice
\begin{align} \label{Lax-Toda-2}
\psi_{n,x}=q_{n,x}\varphi_n-\varphi_{n+1}, \qquad \varphi_{n,y}=e^{q_n-q_{n-1}}\psi_{n-1}.
\end{align}
Both of these Lax pairs are not symmetric. Let us construct combined, in a sense symmetric, Lax pair for the Toda lattice by choosing the following set of functions
$$
\left\{\psi_{2k-1},\varphi_{2k}\right\}^{+\infty}_{k=-\infty}
$$
as the new set of eigenfunctions. 
We exclude functions $\psi_{2k}$ and $\varphi_{2k-1}$ from the linear equations \eqref{Lax-Toda-1}, \eqref{Lax-Toda-2}  by virtue of the equalities
\begin{align} \label{psi-2k-varphi-2k-1}
\psi_{2k}=e^{-q_{2k}}\varphi_{2k}, \qquad \varphi_{2k-1}=e^{q_{2k-1}}\psi_{2k-1}.
\end{align}
As a result, we arrive at the following overdetermined system of linear equations
\begin{align} 
&\psi_{2n-1,x}=-e^{-q_{2n-1}}\varphi_{2n}, \phantom{q_{2nx}\varphi_{2n}} \qquad \psi_{2n-1,y}=-q_{2n-1,y}\psi_{2n-1}+e^{-q_{2n-2}}\varphi_{2n-2}, \label{Lax-psi-xy} \\
&\varphi_{2n,x}=q_{2n,x}\varphi_{2n}-e^{q_{2n+1}}\psi_{2n+1}, \qquad \varphi_{2n,y}=e^{q_{2n}}\psi_{2n-1}. \label{Lax-varphi-xy} 
\end{align}
The compatibility condition for a pair of equations \eqref{Lax-psi-xy} is equivalent to the relation
\begin{align*} 
q_{2n-1,xy}=e^{q_{2n}-q_{2n-1}}-e^{q_{2n-1}-q_{2n-2}}
\end{align*}
while the compatibility condition of  \eqref{Lax-varphi-xy}  is expressed as
\begin{align*} 
q_{2n,xy}=e^{q_{2n+1}-q_{2n}}-e^{q_{2n}-q_{2n-1}}.
\end{align*}
In other words, \eqref{Lax-psi-xy}, \eqref{Lax-varphi-xy} define an alternative non-autonomous Lax pair for the Toda lattice. 

We make the following substitutions $\xi=-2x$, $\eta=-2y$
$$
\hat{q}_1(n)=-e^{-q_{2n-1}}, \qquad \hat{q}_2(n)=e^{q_{2n}}
$$
in the equations \eqref{Lax-psi-xy}, \eqref{Lax-varphi-xy}. As a result, we get
\begin{align} 
&\psi_{2n-1,\xi}=-\frac{1}{2}\hat{q}_1(n)\varphi_{2n}, \qquad \varphi_{2n,\eta}=-\frac{1}{2}\hat{q}_2(n)\psi_{2n-1}, \label{Lax-xi-eta-1}\\
&\begin{aligned} \label{Lax-xi-eta-2}
&\psi_{2n-1,\eta}=-\frac{\hat{q}_{1,\eta}(n)}{\hat{q}_{1}(n)}\psi_{2n-1}-\frac{1}{2}\hat{q}_{2}(n)\varphi_{2n-2},\\
&\varphi_{2n-2,\xi}=\frac{\hat{q}_{2,\xi}(n-1)}{\hat{q}_{2}(n-1)}\varphi_{2n-2}-\frac{1}{2}\hat{q}_{1}(n)\psi_{2n-1}.
\end{aligned}
\end{align}
The obtained Lax pair of the Toda lattice is in a complete agreement with the Lax pair \eqref{Lax-psi-varphi-t}, \eqref{Lax-psi-varphi-xi-eta} of the system 
\eqref{sym-Toda}, \eqref{nonloc-Toda}. Indeed, the system of equations \eqref{Lax-xi-eta-1} exactly coincides with the Lax equation \eqref{Lax-psi-varphi-xi-eta} for $n=1$.

Let us formally define the system of equations \eqref{Lax-psi-varphi-t} for all integer values of the variable  $n$, assuming that
\begin{align} 
\begin{aligned} \label{form-psi-phi-t}
&i\psi_{2n-1,t}=\psi_{2n-1,\xi\xi}-\psi_{2n-1,\eta\eta}-g_1(n)\psi_{2n-1}+\hat{q}_{1,\xi}(n)\varphi_{2n},\\
&i\varphi_{2n,t}=\varphi_{2n,\xi\xi}-\varphi_{2n,\eta\eta}+g_2(n)\varphi_{2n}-\hat{q}_{2,\eta}(n)\psi_{2n-1}.
\end{aligned}
\end{align}
As a result, we have three linear systems. Let us explain their meaning in the following two theorems.
\begin{theorem}
System of equations \eqref{Lax-xi-eta-1}, \eqref{Lax-xi-eta-2} is compatible if and only if the coefficients $\hat{q}_{1}(n)$, $\hat{q}_{2}(n)$ satisfy the following nonlinear equations
\begin{align} 
\begin{aligned} \label{sys-q1q2-xi-eta}
&\hat{q}_{1,\xi\eta}(n)=\frac{\hat{q}_{1,\xi}(n)\hat{q}_{1,\eta}(n)}{\hat{q}_{1}(n)}-\frac{1}{4}\hat{q}^2_{1}(n)\hat{q}_{2}(n)+\frac{1}{4\hat{q}_{2}(n-1)},\\
&\hat{q}_{2,\xi\eta}(n)=\frac{\hat{q}_{2,\xi}(n)\hat{q}_{2,\eta}(n)}{\hat{q}_{2}(n)}+\frac{1}{4}\hat{q}^2_{2}(n)\hat{q}_{1}(n)-\frac{1}{4\hat{q}_{1}(n+1)}.
\end{aligned}
\end{align}
\end{theorem}
Note that \eqref{sys-q1q2-xi-eta} is nothing else than the Toda lattice \eqref{Toda} rewritten using \eqref{connection}.

\begin{theorem}
System  \eqref{Lax-xi-eta-1}, \eqref{form-psi-phi-t} is consistent if and only if the following relations are valid:
\begin{align} 
&\begin{aligned} \label{sys-q1q2-t}
&i\hat{q}_{1,t}(n)+\hat{q}_{1,\xi\xi}(n)+\hat{q}_{1,\eta\eta}(n)-\hat{q}_{1}(n)(g_1(n)+g_2(n))=0,\\
-&i\hat{q}_{2,t}(n)+\hat{q}_{2,\xi\xi}(n)+\hat{q}_{2,\eta\eta}(n)+\hat{q}_{2}(n)(g_1(n)+g_2(n))=0,
\end{aligned}\\
&\begin{aligned} \label{sys-g1g2-xi-eta}
&\frac{\partial}{\partial \xi}g_1(n)=-\frac{1}{2}\frac{\partial}{\partial \eta}(\hat{q}_{1}(n)\hat{q}_{2}(n)), \\
&\frac{\partial}{\partial \eta}g_2(n)=-\frac{1}{2}\frac{\partial}{\partial \xi}(\hat{q}_{1}(n)\hat{q}_{2}(n)).
\end{aligned}
\end{align}
\end{theorem}

Equations \eqref{sys-q1q2-t}, \eqref{sys-g1g2-xi-eta} define a family of solutions of system \eqref{sym-Toda}, \eqref{nonloc-Toda}, related by the B\"acklund transformation generated by lattice \eqref{Toda}. More precisely, we obtain a family of systems of the form \eqref{sym-Toda}, \eqref{nonloc-Toda}, numbered by the parameter $n$ so that the transition from one system to another (neighboring!) is carried out by means of an invertible B\"acklund transformation \eqref{backlund}, \eqref{nonloc-new} or \eqref{inverse-trans-q}, \eqref{inverse-trans-g}.

\section{Generalized Toda lattices}

It is well known that the Toda lattice admits finite-field reductions with enhanced integrability. These reductions are hyperbolic systems of exponential type whose general solution can be represented explicitly \cite{d15}. Such kind systems are called generalized Toda lattices and are closely related to the Cartan matrices of simple Lie algebras. A discussion of the history of the problem and a detailed exposition of the classical results of Darboux, Goursat, Moutard, and others on this topic can be found in the excellent survey \cite{GTs}.

For example, the generalized Toda lattice corresponding to a simple Lie algebra  $A_N$ is a reduction of the Toda lattice obtained by imposing boundary conditions of the form
\begin{equation}\label{bound-cond}
e^{-q_0}=0, \qquad e^{q_{N+1}}=0.
\end{equation}
It should be noted that the cut-off constraints \eqref{bound-cond} are compatible with higher symmetries of the Toda lattice (see \cite{GurelHab97}), in particular with system \eqref{sys-q1q2-t}, \eqref{sys-g1g2-xi-eta}. Here compatibility is understood in the sense that system \eqref{sys-q1q2-t}, \eqref{sys-g1g2-xi-eta} under the conditions \eqref{bound-cond} turns into a system of equations defined for all $n\in\left[1,N\right]$, commuting with the generalized Toda lattice corresponding to the Cartan matrix of the algebra $A_N$. These two circumstances: existence of explicit formulas expressing the general solution of the generalized Toda lattice and commutativity of two flows under consideration can be used to construct particular solutions of system \eqref{sym-Toda}, \eqref{nonloc-Toda}. To do this, it is sufficient to determine the dependence of explicit solutions of the generalized Toda lattice on time $t$. Below, we illustrate this algorithm for solving  system \eqref{sym-Toda}, \eqref{nonloc-Toda} with an example of the generalized Toda lattice corresponding to the Lie algebra $A_2$.

\section{Solution of the generalized Toda lattice corresponding to the algebra $A_2$}

Let us consider the generalized Toda lattice corresponding to $A_2$
\begin{equation}\label{gen-Toda}
q_{1,xy}=e^{q_2-q_1}, \qquad q_{2,xy}=-e^{q_2-q_1}.
\end{equation}
It is obtained from the Toda lattice by means of truncation conditions
\begin{equation}\label{bound-cond-GT}
e^{-q_j}=0 \quad \mbox{for} \quad j\leqslant 0, \qquad e^{q_j}=0 \quad \mbox{for} \quad j\geqslant 3.
\end{equation}
Various methods for constructing solutions of open chains are known in the literature (see, for example, \cite{GTs}, \cite{LeznovSavbook}, \cite{d10}, \cite{d25}). We will derive general solution to \eqref{gen-Toda} using only elementary reasoning. In doing so, together with the solution of \eqref{gen-Toda}, we will find explicit expressions for the eigenfunctions of the Lax operators.

Let us first derive a Lax pair for \eqref{gen-Toda}. To this end we impose cut-off constraints of the form 
$$
\psi_n\equiv 0 \quad \mbox{for} \quad n\neq 1 \quad \mbox{and} \quad n\neq 2
$$
on system \eqref{Lax-Toda-1}. It leads together with \eqref{bound-cond-GT} to the following Lax pair for \eqref{gen-Toda}:
\begin{equation} 
\begin{aligned} \label{Lax-GT}
&\psi_{1,x}=-e^{q_2-q_1}\psi_2, \qquad \psi_{1,y}=-q_{1,y}\psi_1,\\
&\psi_{2,x}=0, \phantom{e^{q_2-q_1}\psi_2,} \qquad \psi_{2,y}=-q_{2,y}\psi_2+\psi_1.
\end{aligned}
\end{equation}
Now we will show that the resulting system of equations can be integrated explicitly. Let us first exclude the difference of variables $q_1$, $q_2$ and their derivatives from  \eqref{Lax-GT}
\begin{equation}\label{Lax-GT'}
q_2-q_1=\ln \frac{-\psi_{1,x}}{\psi_2}, \qquad q_{2,y}=\frac{\psi_1}{\psi_2}-\left(\ln \psi_2\right)_y, \qquad q_{1,y}=-\left(\ln \psi_1\right)_y.
\end{equation}
Using \eqref{Lax-GT'} it is easy to obtain a system of equations for eigenfunctions
\begin{equation*}
\frac{\psi_{1,xy}}{\psi_{1,x}}=\frac{\psi_{1,y}}{\psi_1}+\frac{\psi_1}{\psi_2}, \qquad \psi_{2,x}=0.
\end{equation*}
By substitution $\psi_1=e^u$ we reduce last system to the form
\begin{equation*}
u_{xy}=\frac{u_xe^u}{\psi_2}, \qquad \psi_2=\psi_2(y).
\end{equation*}
Now integrating both parts of the resulting equation with respect to $x$, we find
\begin{equation*}
u_y=e^u\frac{1}{\psi_2(y)}+S(y),
\end{equation*}
where $S(y)$ is an arbitrary function. Next, assuming $z=e^{-u}$ we reduce the latter to a linear inhomogeneous equation of the form
\begin{equation}\label{zy}
z_y=-zS(y)-\frac{1}{\psi_2(y)}.
\end{equation}
Let us introduce a new function $\delta(y)$ by setting $\delta'(y)=S(y)$ and make the following change of variables in \eqref{zy}: $z=Ce^{-\delta(y)}$. As a result, we obtain
\begin{equation*}
\frac{\partial}{\partial y}C=-\frac{1}{\psi_2(y)}e^{\delta(y)}.
\end{equation*}
Thus obviously we have $\frac{\partial^2}{\partial x\partial y}C=0$, therefore
\begin{equation}\label{C}
C(x,y)=\theta(y)+\rho(x).
\end{equation}
Here $\theta(y)$ and $\rho(x)$ are arbitrary smooth functions, such that $\theta'(y)=-\frac{e^{\delta(y)}}{\psi_2(y)}$. This fact obviously implies explicit representations for $\psi_2(y)$ and $z$
\begin{equation}\label{psi-2}
\psi_2(y)=-\frac{e^{\delta(y)}}{\theta'(y)}, \qquad z=(\theta(y)+\rho(x))e^{-\delta(y)}.
\end{equation}
Afterwards, we use relation $\psi_1=e^u=\frac{1}{z}$ and get
\begin{equation}\label{psi-1}
\psi_1=\frac{e^{\delta(y)}}{\theta(y)+\rho(x)}.  
\end{equation}
From equations \eqref{Lax-GT'}, \eqref{psi-2} and  \eqref{psi-1}, one can easily  obtain explicit representations for the remaining desired functions
\begin{equation}\label{parametrization}
\begin{aligned}
&\hat{q}_1=-e^{-q_1}=-\frac{e^{\delta(y)-\alpha(x)}}{\theta(y)+\rho(x)},\\
&\hat{q}_2=e^{q_2}=-\frac{e^{-\delta(y)+\alpha(x)}\theta'(y)\rho'(x)}{\theta(y)+\rho(x)},\\
&\varphi_2=e^{q_2}\psi_2=\frac{e^{\alpha(x)}\rho'(x)}{\theta(y)+\rho(x)},
\end{aligned}
\end{equation}
where $\delta(y)$, $\theta(y)$, $\alpha(x)$, $\rho(x)$ are arbitrary functions. Below we define the dynamics of these four functions with respect to $t$ in a suitable way. The dynamics should ensure that $\hat{q}_1$, $\hat{q}_2$, $g_1$, $g_2$ satisfy \eqref{sym-Toda}, \eqref{nonloc-Toda}.

\subsection{Finding out the dependence of functional parameters on time $t$} 

The next step of our algorithm of finding a solution to  \eqref{sym-Toda}, \eqref{nonloc-Toda}  is to determine the dependence of the functional parameters $\delta(y)$, $\theta(y)$, $\alpha(x)$, $\rho(x)$ on time. For this purpose, we use the system of equations
\begin{equation}\label{eq-t}
\begin{aligned}
&i\psi_{1,t}=\psi_{1,\xi\xi}-\psi_{1,\eta\eta}-g_1\psi_1+\hat{q}_{1,\xi}\varphi_2,\\
&i\varphi_{2,t}=\varphi_{2,\xi\xi}-\varphi_{2,\eta\eta}+g_2\varphi_2-\hat{q}_{2,\eta}\psi_1,\\
&\hat{q}_{1,t}=i\left(\hat{q}_{1,\xi\xi}+\hat{q}_{1,\eta\eta}\right)+i\hat{q}_{1}(g_1+g_2),\\
&\hat{q}_{2,t}=-i\left(\hat{q}_{2,\xi\xi}+\hat{q}_{2,\eta\eta}\right)-i\hat{q}_{2}(g_1+g_2)
\end{aligned}
\end{equation}
that follows from the dynamical system itself and its Lax pair.

Let us rewrite the formulas above (see \eqref{psi-1}, \eqref{parametrization}) that define a parametrization of the unknown functions in terms of the appropriate variables $\xi=-2x$, $\eta=-2y$
\begin{equation}\label{sol-xi-eta}
\begin{aligned}
&\psi_1=\frac{e^{\bar{\delta}(\eta)}}{\bar{\theta}(\eta)+\bar{\rho}(\xi)}, \qquad \varphi_2=-2\frac{e^{\bar{\alpha}(\xi)}\bar\rho'(\xi)}{\bar{\theta}(\eta)+\bar{\rho}(\xi)},\\
&\hat{q}_1=-\frac{e^{\bar{\delta}(\eta)-\bar{\alpha}(\xi)}}{\bar{\theta}(\eta)+\bar{\rho}(\xi)},\qquad \hat{q}_2=-4\frac{e^{-\bar{\delta}(\eta)+\bar{\alpha}(\xi)}\bar{\theta}'(\eta)\bar{\rho}'(\xi)}{\bar{\theta}(\eta)+\bar{\rho}(\xi)},
\end{aligned}
\end{equation}
here $\bar \delta(\eta)=\delta(y)$, $\bar \theta(\eta)=\theta(y)$, $\bar{\alpha}(\xi)=\alpha(x)$, $\bar{\rho}(\xi)=\rho(x)$.

Note that in the example under consideration the equations for nonlocalities
\begin{equation*}
\frac{\partial}{\partial \xi}g_1=-\frac{1}{2}\frac{\partial}{\partial \eta}(\hat{q}_{1}\hat{q}_{2}), \qquad
\frac{\partial}{\partial \eta}g_2=-\frac{1}{2}\frac{\partial}{\partial \xi}(\hat{q}_{1}\hat{q}_{2})
\end{equation*}
can be integrated explicitly. Assuming the integration constants to be  zero, we obtain
\begin{equation}\label{sol-g1g2-xi-eta}
\begin{aligned}
&g_1=\frac{2\bar{\theta}''(\eta)}{\bar{\theta}(\eta)+\bar{\rho}(\xi)}-\frac{2(\bar{\theta}'(\eta))^2}{(\bar{\theta}(\eta)+\bar{\rho}(\xi))^2},\\
&g_2=\frac{2\rho''(\xi)}{\bar{\theta}(\eta)+\bar{\rho}(\xi)}-\frac{2(\rho'(\xi))^2}{(\bar{\theta}(\eta)+\bar{\rho}(\xi))^2}.
\end{aligned}
\end{equation}
Let us substitute explicit expressions for the functions
$\psi_1$, $\varphi_2$, $\hat{q}_1$, $\hat{q}_2$, $g_1$, $g_2$ into \eqref{eq-t}. As a result, we obtain a system of equations, which after some simplifications takes the following form (for convenience, in the equations below we omit the bar over the letter):
\begin{equation}\label{eq-coeff-t}
\begin{aligned}
&1) \quad \alpha_t=i\alpha_{\xi\xi}-i\alpha^2_\xi,\\
&2) \quad \delta_t=i\delta_{\eta\eta}+i\delta^2_\eta+s_1(t),\\
&3) \quad \rho_t=-i\rho_{\xi\xi}-2i\alpha_\xi\rho_\xi+s_1(t)\rho+s_2(t),\\
&4) \quad \theta_t=-i\theta_{\eta\eta}+2i\delta_\eta \theta_\eta+s_1(t)\theta-s_2(t),
\end{aligned}
\end{equation}
where $s_1(t)$ and $s_2(t)$ are arbitrary functions.

An elementary analysis of the obtained equations convinces us that the following statement is true.

\begin{theorem}\label{Th3}
System \eqref{eq-coeff-t} is linearized by simple Cole--Hopf type substitutions and is reduced to a system of four heat  equations with imaginary time
\begin{equation}\label{beta-t}
\beta_{1,t}=i\beta_{1,\xi\xi}, \qquad \beta_{2,t}=i\beta_{2,\eta\eta}, \qquad \beta_{3,t}=-i\beta_{3,\xi\xi}, \qquad \beta_{4,t}=-i\beta_{4,\eta\eta}.
\end{equation}
\end{theorem}

\begin{proof}
Let us begin with the first equation in \eqref{eq-coeff-t}. Evidently substitution $\alpha=-\ln \beta_1$  reduces it to the form $\beta_{1,t}=i\beta_{1,\xi\xi}$.
The second equation is reduced to the form
$$\chi_t=i\chi_{\eta\eta}+s_1(t)\chi$$
by a similar transformation $\delta=\ln \chi$.
Afterwards, setting
$\chi=\beta_2e^{\int\limits^{t}_{t_0}s_1(\tau) d \tau},$
we obtain the desired form
$\beta_{2,t}=i\beta_{2,\eta\eta}.$

Let us proceed with the third equation.  First, we differentiate it with respect to $\xi$, and then we make the substitution $\rho_\xi=\sigma\beta_1$, where $\beta_1$ was already defined in the previous step. Then we simplify by means of the relation $\beta_{1,t}=i\beta_{1,\xi\xi}$ and get 
$$\sigma_t=-i\sigma_{\xi\xi}+s_1(t)\sigma.$$
The latter is reduced to the heat equation 
$\beta_{3,t}=-i\beta_{3,\xi\xi}$
by point transformation $$\sigma=\beta_3e^{\int\limits^{t}_{t_0}s_1(\tau) d \tau}.$$

In a similar way we investigate  $4)$. We differentiate it with respect to $\eta$ and make a substitution $\theta_{\eta}=\chi\beta_4,$ $\chi=e^\delta$. Then due to relation $\chi_t=i\chi_{\eta\eta}+s_1(t)\chi$ we obtain that 
$$\beta_{4,t}=-i\beta_{4,\eta\eta}.$$
Theorem \ref{Th3} is proved.
\end{proof}

Now it remains to find a parametrization of the desired functions in terms of the functions $\beta_1$--$\beta_4$. Obviously, it follows from the  reasoning above that
\begin{equation}\label{alpha-H}
\alpha=-\ln \beta_1, \qquad \delta=\ln \beta_2+\int\limits^{t}_{0}s_1(\tau) d \tau.
\end{equation}
Here for the simplicity we assume that $s_1(0)=0$. Let us solve equations $\rho_\xi=\sigma\beta_1$ and $\theta_\eta=\chi\beta_4$, assuming that $\rho(0,t)=A(t)$, $\theta(0,t)=B(t)$. As a result we find
\begin{equation}\label{sol-rho-w}
\begin{aligned}
&\rho(\xi,t)=e^{\int\limits^{t}_{0}s_1(\tau) d \tau}\left(\int\limits^{\xi}_{0}\beta_1(\zeta,t)\beta_3(\zeta,t) d \zeta+A(t)\right),\\
&\theta(\eta,t)=e^{\int\limits^{t}_{0}s_1(\tau) d \tau}\left(\int\limits^{\eta}_{0}\beta_2(\nu,t)\beta_4(\nu,t) d \nu+B(t)\right).
\end{aligned}
\end{equation}
We substitute the obtained relations into \eqref{sol-xi-eta} and get
\begin{equation}\label{sol-xi-eta-t}
\hat{q}_1=-\frac{\beta_1(\xi,t)\beta_2(\eta,t)}{R},\qquad \hat{q}_2=-\frac{4\beta_3(\xi,t)\beta_4(\eta,t)}{R}.
\end{equation}
Here
\begin{align*}
R=\int\limits^{\xi}_{0}\beta_1(\zeta,t)\beta_3(\zeta,t) d \zeta+\int\limits^{\eta}_{0}\beta_2(\nu,t)\beta_4(\nu,t) d \nu+C(t),
\end{align*}
where $C(t)=A(t)+B(t)$. For the non-local variables we have explicit relations
\begin{equation}\label{sol-g1g2-xi-eta-t}
\begin{aligned}
&g_1=\frac{2(\beta_2(\eta,t)\beta_4(\eta,t))_{\eta}}{R}-\frac{2(\beta_2(\eta,t)\beta_4(\eta,t))^2}{(R)^2},\\
&g_2=\frac{2(\beta_1(\xi,t)\beta_3(\xi,t))_{\xi}}{R}-\frac{2(\beta_1(\xi,t)\beta_3(\xi,t))^2}{(R)^2}.
\end{aligned}
\end{equation}
Let us summarize the discussion in this section with the following statement, which is easily proved by a direct computation.
\begin{theorem} Functions $\hat{q}_1$, $\hat{q}_2$, $g_1$ and $g_2$ defined by equalities \eqref{sol-xi-eta-t}, \eqref{sol-g1g2-xi-eta-t} determine a solution of system \eqref{sym-Toda}, \eqref{nonloc-Toda} if and only if 
$$\frac{\partial}{\partial t}C(t)=0.$$
Note that the solution depends on four arbitrary functions $\beta_1(\xi,t)$, $\beta_2(\eta,t)$, $\beta_3(\xi,t)$, $\beta_4(\eta,t)$, which satisfy the heat  equations \eqref{beta-t}. 
\end{theorem}

\section{Solution of the Davey--Stewartson I equation}

We will show that the above solution, with an appropriate choice of functional parameters, admits a complex-conjugate reduction that reduces the solution to a solution of the Davey--Stewartson I equation.

In order to make the formulas above more symmetrical, we introduce new notations
$$\beta_3=\frac{1}{2}\tilde{\beta}_3, \qquad \beta_4=\frac{1}{2}\tilde{\beta}_4.$$
Then formulas \eqref{sol-xi-eta-t}, \eqref{sol-g1g2-xi-eta-t} take the form
\begin{equation}\label{sol-DS}
\begin{aligned}
&\hat{q}_1=-\frac{\beta_1\beta_2}{R}, \qquad \hat{q}_2=-\frac{\tilde{\beta}_3\tilde{\beta}_4}{R},\\
&R=\int\limits^{\xi}_{0}\frac{\beta_1\tilde{\beta}_3}{2} d \zeta+\int\limits^{\eta}_{0}\frac{\beta_2\tilde{\beta}_4}{2} d \nu+C,\\
&g_1=\frac{(\beta_2\tilde{\beta}_4)_{\eta}}{R}-\frac{(\beta_2\tilde{\beta}_4)^2}{2(R)^2},\\
&g_2=\frac{(\beta_1\tilde{\beta}_3)_{\xi}}{R}-\frac{(\beta_1\tilde{\beta}_3)^2}{2(R)^2}.
\end{aligned}
\end{equation}
It is  clear now that the formulas above are consistent with the reduction 
\begin{equation}\label{beta-complex}
\tilde{\beta}_3=\beta^*_1, \qquad \tilde{\beta}_4=\beta^*_2, \qquad C^*=C,
\end{equation}
where the star denotes complex conjugation. It is assumed that $C=\const$. In this case we have
\begin{equation}\label{R}
R=\int\limits^{\xi}_{0}\frac{|\beta_1(\zeta,t)|^2}{2} d \zeta+\int\limits^{\eta}_{0}\frac{|\beta_2(\nu,t)|^2}{2} d \nu+C=R^*.
\end{equation}
Therefore the relations hold $\hat{q}^*_1=\hat{q}_2$, $g_1^*=g_1$, $g_2^*=g_2$ and the system of equations \eqref{sym-Toda}, \eqref{nonloc-Toda} imply that
\begin{align}\label{ku1}
&i\hat{q}_{1,t}+\hat{q}_{1,\xi\xi}+\hat{q}_{1,\eta\eta}+\hat{q}_{1}(g_1+g_2)=0,\\
&\frac{\partial^2}{\partial\xi\partial\eta}g_1=-\frac{1}{2}\frac{\partial^2}{\partial\eta^2}|\hat{q}_1|^2, \qquad
\frac{\partial^2}{\partial\xi\partial\eta}g_2=-\frac{1}{2}\frac{\partial^2}{\partial\xi^2}|\hat{q}_1|^2. \nonumber
\end{align}
Adding the last two equalities term by term, we obtain
\begin{equation}\label{g1g2}
\frac{\partial^2}{\partial\xi\partial\eta} \left(g_1+g_2\right)=-\frac{1}{2}\left(\frac{\partial^2}{\partial\eta^2}+\frac{\partial^2}{\partial\xi^2}\right)|\hat{q}_1|^2.
\end{equation}
Let us introduce the notations $u=\hat{q}_1$, $v=g_1+g_2$ and represent the relations \eqref{ku1}, \eqref{g1g2} in the following form
\begin{align*}
&iu_t+\Delta u+uv=0,\\
&v_{\xi\eta}+\frac{1}{2}\Delta |u|^2=0,
\end{align*}
where $\Delta$ stands for the Laplace operator
$$\Delta=\frac{\partial^2}{\partial\xi^2}+\frac{\partial^2}{\partial\eta^2}.$$

In other words, Theorem 5.2 implies the following statement:

\begin{corollary} Functions   
\begin{equation}\label{DS-sol}
\begin{aligned}
&u=-\frac{\beta_1(\xi,t) \beta_2(\eta,t)}{R(\xi,\eta,t)},\\
&v=\frac{1}{R}\left(\frac{\partial}{\partial\xi}|\beta_1(\xi,t) |^2+\frac{\partial}{\partial\eta}|\beta_2(\eta,t) |^2   \right)-\frac{1}{2R^2}\left(|\beta_1(\xi,t) |^4+|\beta_2(\eta,t) |^4   \right)
\end{aligned}
\end{equation}
determine a solution to the Davey--Stewartson equation \eqref{eq-DSI}. The solution depends on arbitrary solutions   $\beta_1(\xi,t)$, $\beta_2(\eta,t)$ of the heat equation with imaginary time (see \eqref{beta-t}), function $R(\xi,\eta,t)$ is given by \eqref{R}.
\end{corollary}

\section{An illustrative example}

In this section we construct a solution of the Davey--Stewartson I equation using the formulas \eqref{DS-sol}, \eqref{beta-t}. To this end we have to solve the Cauchy problem for the equation
\begin{equation*}
\beta_{1,t}=i\beta_{1,\xi\xi}
\end{equation*}
with some initial data $\beta_1(\xi,0)=\beta_0(\xi)$. Solution to the problem is given by the Poisson formula
\begin{equation}\label{Poisson}
\beta_1(\xi,t)=\frac{1}{2\sqrt{it}}\int\limits_{-\infty}^{\infty}\beta_0(\xi)e^{-\frac{(\xi-\zeta)^2}{4it}}d\zeta.
\end{equation}
It is remarkable that for some special choices of the initial data integral \eqref{Poisson} is evaluated in an explicit form. Let us take $\beta_0(\xi)=e^{-\xi^2}$, then we obtain
\begin{equation*}
\beta_1(\xi,t)=\frac{1}{\sqrt{4it+1}}e^{-\frac{\xi^2}{4it+1}}.
\end{equation*}
The solution to the Cauchy problem $\beta_{2,t}=i\beta_{2,\eta\eta}$, $\beta_2(\eta,0)=e^{-\eta^2}$,
has a similar form, i.e. we find
\begin{equation*}
\beta_2(\eta,t)=\frac{1}{\sqrt{4it+1}}e^{-\frac{\eta^2}{4it+1}}.
\end{equation*}
Then the solution $u(\xi,\eta,t)$ of the Davey--Stewartson equation is given by formula
\begin{equation*}
u(\xi,\eta,t)=-\frac{1}{(4it+1)R}e^{-\frac{\xi^2+\eta^2}{4it+1}},
\end{equation*}
where
\begin{equation*}
R=-\frac{1}{\sqrt{16t^2+1}}\left(\int\limits_{0}^{\xi}e^{-\frac{2\zeta^2}{16t^2+1}}d\zeta+\int\limits_{0}^{\eta}e^{-\frac{2\nu^2}{16t^2+1}}d\nu\right) +C.
\end{equation*}
In this case, the non-local variable $v(\xi,\eta,t)$ has the form:
\begin{align*}
v(\xi,\eta,t)=\frac{8}{(16t^2+1)^{\frac{3}{2}}R}\left(\xi e^{-\frac{2\xi^2}{16t^2+1}}+\eta e^{-\frac{2\eta^2}{16t^2+1}}\right)
-\frac{2}{(16t^2+1)R^2}\left(e^{-\frac{4\xi^2}{16t^2+1}}+ e^{-\frac{4\eta^2}{16t^2+1}}\right).
\end{align*}
Figures 1 and 2 show graphs of the found solution for specific values of the variable $t$ and the constant $C$.
\begin{figure}[h!]
\centering
\begin{minipage}{.5\textwidth}
  \centering
  \includegraphics[width=.5\linewidth]{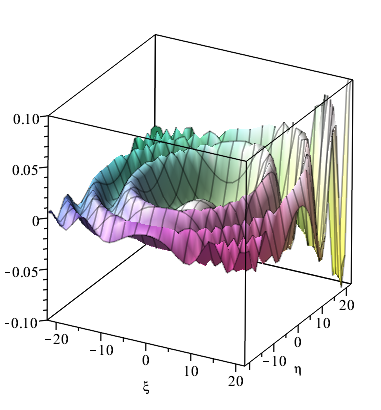}
  \caption{Re(u), t=10, C=1.}
\end{minipage}%
\begin{minipage}{.51\textwidth}
  \centering
  \includegraphics[width=.51\linewidth]{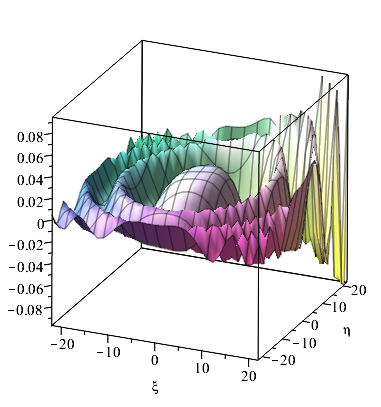}
  \caption{Im(u), t=10, C=1.}
\end{minipage}
\end{figure}

\section*{Conclusion}
In the paper, an algorithm for constructing explicit analytic solutions of the Davey--Stewartson type equations  is developed. It uses a two-dimensional dressing chain. The algorithm is based on the use of the reductions of lattices consistent with higher symmetries and Lax pairs. 
The article examines in detail a specific example where the generalized Toda lattice corresponding to the simple Lie algebra $A_2$ is used as a dressing chain.

 Apparently, the following conjecture deserves attention.
\begin{hypothesis} For arbitrary $N$, the general solution of the open Toda chain associated with the simple Lie algebra $A_N$ generates an explicit solution of the equation \eqref{eq-DSI}.
\end{hypothesis}

Исмагил Талгатович Хабибуллин (ответственный за переписку)\\
Институт математики с ВЦ УФИЦ РАН,\\ ул. Чернышевского, 112,\\ 450008, г.Уфа, Россия \\ 
электронная почта {habibullinismagil@gmail.com}\\

Айгуль Ринатовна Хакимова\\
Институт математики с ВЦ УФИЦ РАН,\\ ул. Чернышевского, 112,\\ 450008, г.Уфа, Россия \\ 
электронная почта {aigul.khakimova@mail.ru}\\

\end{document}